\documentclass[12pt]{article}
\usepackage{amsmath,amsthm}
\usepackage{graphicx}
\usepackage{url} 
\usepackage{algorithmic}
\usepackage[ruled]{algorithm2e} 
\usepackage{setspace}
\usepackage{bbm}
\usepackage{natbib}
\usepackage[pdftex,pagebackref=false]{hyperref}

\usepackage{graphics}

\newcommand{\blind}{0}

\newtheorem{proposition}{Proposition}

\begin{document}

\def\spacingset#1{\renewcommand{\baselinestretch}%
{#1}\small\normalsize} \spacingset{1}


\if0\blind
{
  \title{\bf Particle Gibbs for Likelihood-Free Inference of State Space Models with Application to Stochastic Volatility}
  \author{Zhaoran Hou and Samuel W.K. Wong\thanks{Author for correspondence: samuel.wong@uwaterloo.ca}\hspace{.2cm}\\
    Department of Statistics and Actuarial Science, University of Waterloo}
  \maketitle
} \fi

\if1\blind
{
  \bigskip
  \bigskip
  \bigskip
  \begin{center}
    {\LARGE\bf Title}
\end{center}
  \medskip
} \fi

\bigskip
\begin{abstract}
State space models (SSMs) are widely used to describe dynamic systems. However, when the likelihood of the observations is intractable, parameter inference for SSMs cannot be easily carried out using standard Markov chain Monte Carlo or sequential Monte Carlo methods. In this paper, we propose a particle Gibbs sampler as a general strategy to handle SSMs with intractable likelihoods in the approximate Bayesian computation (ABC) setting. The proposed sampler incorporates a conditional auxiliary particle filter, which can help mitigate the weight degeneracy often encountered in ABC. To illustrate the methodology, we focus on a classic stochastic volatility model (SVM) used in finance and econometrics for analyzing and interpreting volatility. Simulation studies demonstrate the accuracy of our sampler for SVM parameter inference, compared to existing particle Gibbs samplers based on the conditional bootstrap filter. As a real data application, we apply the proposed sampler for fitting an SVM to S\&P 500 Index time-series data during the 2008 financial crisis.
\end{abstract}

\vfill

\newpage
\spacingset{1} 
\section{Introduction}
\label{sec:intro}

State space models (SSMs) describe dynamic systems via a time series of unobserved variables and observations generated conditional on those variables \citep{kitagawa1998self}. The unobserved variables are commonly known as \textit{hidden states}, and are often taken to be a discrete-time Markov process; then the SSM is specified via the transition probability of the hidden states and the likelihood of the observations. When the model also involves unknown parameters, Markov chain Monte Carlo \citep[MCMC,][]{jacquier1994bayesian} and sequential Monte Carlo \citep[SMC,][]{liu2001combined,storvik2002particle,carvalho2010particle} are general computational methods for inferring the parameters and hidden states together, by drawing samples from their joint posterior distribution in a Bayesian setting. However, the inference problem is more difficult when the likelihoods are intractable: examples are stochastic kinetic models in systems biology \citep{owen2015scalable,lowe2023accelerating} and stochastic volatility models with an intractable noise term \citep{vankov2019filtering}. In this likelihood-free setting, standard inference methods are no longer applicable, and this paper develops a new sampler which we illustrate in the context of stochastic volatility modeling.

In the analysis of financial time series, volatility is commonly used to quantify uncertainty or risk. 
Given a time series of observed prices $\{P_t\}_{t=0}^T$, for $t=1, \ldots, T$ the \textit{return} is defined as $r_t \equiv \log(P_t) - \log(P_{t-1})$ and the \textit{volatility} is defined as $h_t \equiv Var(r_t | r_{1:t-1})$. A classic SVM \citep{jacquier1994bayesian,jacquier2004bayesian,taylor2008modelling}, which has been widely used in option pricing and portfolio management, adopts the form of an SSM: for $t=1,\dots, T$,
\begin{align}
    \log(h_t) = \tau + \phi \log(h_{t-1}) + \sigma_h \epsilon_t, \quad
    r_t = \sqrt{h_t} \times Z_t
    \label{SVM(1)}
\end{align}
where $\theta \equiv (\tau, \phi, \sigma_h^2)$ denotes the parameters, $Z_t$ and $\epsilon_t$ are independent noise terms, and an initial distribution is assigned for $h_0$. 
In the log-volatility process, $\tau$ contributes to the mean, $\phi$ is an auto-regressive parameter that governs the persistence (with $|\phi| < 1$ required for second-order stationarity) \citep{jacquier2004bayesian}, and $\sigma_h^2$ is a noise level that governs the amount of autocorrelation. Model \eqref{SVM(1)} has a hierarchical structure: we let $l_t(r_t \mid h_t)$ denote the likelihood density, and $g_t(h_t \mid h_{t-1}, \theta)$ denote the transition density.

In practice, $Z_t$ and $\epsilon_t$ in Model \eqref{SVM(1)} are often assumed to follow standard Gaussian distributions. Then we can integrate out $h_{0:T}$ analytically to directly obtain the likelihood $p(r_{1:T} \mid \theta)$, and the SVM parameters can be estimated by standard methods, e.g., maximum likelihood in a frequentist approach \citep{fridman1998maximum}, and  MCMC in a Bayesian approach \citep{jacquier1994bayesian}. However, empirical studies suggest that the Gaussian assumption for $Z_t$ may not adequately capture the heavy tails and skewness of financial time series \citep{engle2001good}; as a more flexible alternative, the so-called \textit{stable distribution} \citep{mandelbrot1963variation,nolan1997numerical} can be adopted for $Z_t$ instead  \citep{lombardi2009indirect,vankov2019filtering}. The stable distribution does not have a closed-form density function, hence employing it for $Z_t$ leads to an analytically intractable likelihood $l_t(r_t \mid h_t)$ that is expensive to evaluate; this is the setting of interest in this paper.

Particle Markov chain Monte Carlo \citep[PMCMC,][]{andrieu2010particle} is a general sampling approach that can be adapted to handle Model \eqref{SVM(1)} in a Bayesian setting with intractable likelihoods. Initially proposed for inference of SSMs, PMCMC provides a class of algorithms that combine the features of MCMC and SMC; in this context, SMC methods (also known as particle filters) are well-suited for drawing posterior samples of the hidden states in an SSM (e.g., $h_{0:T}$ in Model \eqref{SVM(1)}). To our best knowledge, PMCMC is the main type of approach for SVM parameter inference when it is not possible to integrate out $h_{0:T}$ \citep{vankov2019filtering}. Letting $p(\theta)$ denote the joint prior for $\theta$, PMCMC then targets the joint posterior of $\theta$ and $h_{0:T}$, namely
$$
p(\theta, h_{0:T} \mid r_{1:T}) \propto p(\theta)g_0(h_0 \mid \theta) \prod_{t=1}^T g_t(h_t \mid h_{t-1}, \theta)l_t(r_t \mid h_t),
$$
if $l_t(r_t \mid h_t)$ can be evaluated. (The corresponding version of PMCMC that bypasses this likelihood evaluation is discussed in Section \ref{sec:methodology}.)
When closed-form conditional distributions of the parameters are available, a special case of PMCMC, known as \textit{particle Gibbs}, can be implemented. As applied here, the basic strategy of particle Gibbs alternates between sampling from $p(\theta \mid h_{0:T}, r_{1:T})$ and $p(h_{0:T} \mid \theta, r_{1:T})$ at each iteration. With the help of conjugate priors, sampling from $p(\theta \mid h_{0:T}, r_{1:T})$ can be straightforward. Sampling from $p(h_{0:T} \mid \theta, r_{1:T})$ can potentially be handled by SMC as for an SSM; e.g., one might apply SMC and obtain a particle approximation of $p(h_{0:T} \mid \theta, r_{1:T})$, denoted by $\widehat{p}_{SMC}(h_{0:T} \mid \theta, r_{1:T})$. However, it is not valid to simply substitute sampling from $p(h_{0:T} \mid \theta, r_{1:T})$ with sampling from $\widehat{p}_{SMC}(h_{0:T} \mid \theta, r_{1:T})$ in particle Gibbs, because doing so does not admit the target distribution as invariant \citep{andrieu2010particle}. To correctly embed SMC within a particle Gibbs algorithm, the form of SMC known as \textit{conditional SMC} (cSMC) should be implemented instead; we review cSMC algorithms in Section \ref{sec:csmc}.

Approximate Bayesian computation (ABC) is a general technique that can be used to bypass the evaluation of an intractable or expensive likelihood, if one can directly simulate observations from the likelihood \citep{marin2012approximate}. Thus, to perform inference on Model \eqref{SVM(1)} with an analytically intractable $l_t(r_t \mid h_t)$, ABC may be combined with cSMC. The basic idea is to introduce a sequence of auxiliary observations (that are sampled from the likelihood) and then assign weights according to the ``distances'' between the auxiliary observations and the actual observations; ABC-based methods are further reviewed in Section \ref{sec:abc}. It is straightforward to construct ABC versions of existing cSMC algorithms based on the bootstrap filter \citep[BF,][]{gordon1993novel}; however, in practice the particles they generate can tend to have many near-zero weights due to large ``distances'', i.e., the algorithms suffer from weight degeneracy. In the SMC literature, specific techniques to mitigate weight degeneracy include the SMC sampler with annealed importance sampling \citep{del2006sequential}, weight tempering via lookahead strategies \citep{lin2013lookahead}, and drawing multiple descendants per particle \citep{Hou2024-rk}; however, as cSMC must be embedded within every iteration of particle Gibbs, these techniques would be computationally expensive to apply. In contrast, the auxiliary particle filter \citep[APF,][]{pitt1999filtering} is a common alternative to the BF that can help reduce weight degeneracy, since the resampling step of the APF accounts for the one-step-ahead observation.

In this paper, we use the APF as a strategy to reduce weight degeneracy and improve parameter estimation in the particle Gibbs and ABC setting. In related work, \cite{vankov2019filtering} proposed an ABC-based APF that uses auxiliary observations to assign importance weights for PMCMC. However, that ABC-based APF has only been embedded within a particle Metropolis-Hastings algorithm; a Gibbs sampler is preferred over Metropolis-Hastings proposals when closed-form conditional distributions are available for the parameters. Therefore, as the main contribution of this paper, we propose to embed an ABC-based APF within a particle Gibbs algorithm. We show that our proposed algorithm satisfies the form of cSMC, and thus admits the target posterior distribution as invariant. We then perform inference on Model \eqref{SVM(1)} when $Z_t$ is assumed to follow the stable distribution. The results demonstrate that our particle Gibbs algorithm significantly outperforms existing ones, thereby providing a useful computational method for handling SSMs with intractable likelihoods.

The paper is organized as follows: in Section \ref{sec:methodology}, we review ABC and particle Gibbs methods, and present the proposed ABC-based particle Gibbs algorithm with a conditional auxiliary particle filter (ABC-PG-cAPF). In Section \ref{sec:simulation}, we illustrate the effectiveness of the proposed sampler for inference of Model \eqref{SVM(1)} with the stable distribution. In Section \ref{section:application}, we present a real data application by fitting an SVM to S\&P 500 daily returns during the 2008--2009 financial crisis. In Section \ref{section: conclusion}, we briefly summarize the paper and its contributions and discuss some potential future directions.

\section{Methodology}\label{sec:methodology}
\subsection{Model setup}\label{sec:setup}

The \textit{stable distribution} is defined via its characteristic function \citep{nolan2020univariate}: we denote $Z \sim SD(\alpha, \beta, \gamma, \delta)$ if $Z$ follows a stable distribution with parameters  $\alpha \in (0,2]$ for heavy-tailedness, $\beta \in [-1, 1]$ for skewness, $\gamma \in [0, \infty)$ for scale and $\delta \in (-\infty, \infty)$ for location, and $Z$ has characteristic function
$$
\psi_Z(t)=\left\{\begin{aligned} & \exp \left(-\gamma^{\alpha}|t|^{\alpha}\left[1+i \beta \operatorname{sign}(t) \tan \left(\frac{\pi\alpha}{2}\right) \left\{(\gamma|t|)^{1-\alpha}-1\right\}\right]+i \delta t\right) \quad \text { if } \alpha \neq 1 \\ \\ & \exp \left(-\gamma|t|\left\{1+\frac{2}{\pi}i \beta \operatorname{sign}(t) \log (\gamma|t|)\right\}+i \delta t\right) \quad \text { if } \alpha=1. \end{aligned}\right.
$$
Its corresponding density function, $f_Z(z) = \frac{1}{2\pi} \int_{-\infty}^\infty \psi_Z(t) \exp(-izt)dt$, does not have an analytical form in general; some basic facts are that if $\alpha > 1$, $E[Z] = \delta - \beta \gamma\tan(\pi\alpha/2)$ (and undefined otherwise); if $\alpha < 2$, $Var(Z) = \infty$ \citep{mandelbrot1963variation,nolan1997numerical,nolan2020univariate}. Evaluation of $f_Z(z)$ is expensive and commonly needs to be approximated by the fast Fourier transform \citep{mittnik1999computing} or numerical integration \citep{nolan1997numerical,nolan1999algorithm}, but it is feasible to simulate realizations of $Z$ according to $\psi_Z(t)$ based on the work of \cite{kanter1975stable} and \cite{chambers1976method}; in this paper, we use the R package \textit{stabledist} \citep{R:stabledist} for generating random draws of $Z$.

The main model of interest in this paper is then the SVM specified in \eqref{SVM(1)} with $\epsilon_t \sim N(0,1)$ and $Z_t \sim SD(\alpha, \beta, \gamma, \delta)$ for $t = 1, \ldots, T$, all independent. We also set $h_0 \sim \text{Log-normal}(\tau/(1-\phi), \sigma_h^2/(1-\phi^2))$ according to the stationarity of the log-volatility process. The goal is to infer the parameters $\theta = (\tau, \phi, \sigma_h^2)$ given a time series of returns $r_t$. We take a Bayesian approach to inference and assign conjugate priors for the parameters: 
$$
\sigma_h^2 \sim IG(a_0, b_0), \quad
    (\tau, \phi)\mid \sigma_h^2 \sim N(\boldsymbol{\mu}_0, \sigma_h^2 \boldsymbol{\Lambda}_0^{-1}) \text{ with } |\phi| < 1
$$
i.e., the joint prior for $\theta$ is a truncated normal--inverse Gamma with hyperparameters $a_0, b_0, \boldsymbol{\mu}_0, \boldsymbol{\Lambda}_0$, which we denote as $NIG(a_0, b_0, \boldsymbol{\mu}_0, \boldsymbol{\Lambda}_0)$; the truncated support for $\phi$ ensures the required second-order stationarity of the log-volatility process. The conjugacy of these priors follows from the results of Bayesian linear regression, i.e., by defining
$$
\mathbf{X} = \begin{bmatrix}
1 & \log(h_0)\\
\vdots & \vdots\\
1 & \log(h_{T-1})
\end{bmatrix} \text{ and  }
\boldsymbol{y} = \begin{bmatrix}
\log(h_{1})\\
\vdots\\
\log(h_{T})
\end{bmatrix},
$$
then $p(\theta \mid h_{0:T}, r_{1:T}) = p(\theta \mid h_{0:T})$ is a truncated $NIG(a_T, b_T, \boldsymbol{\mu}_T, \boldsymbol{\Lambda}_T)$ with $|\phi| < 1$, where
$$
\begin{aligned} 
\boldsymbol{\Lambda}_{T} &=\left(\mathbf{X}^{\mathrm{T}} \mathbf{X}+\mathbf{\Lambda}_{0}\right), \quad
\boldsymbol{\mu}_{T} =\boldsymbol{\Lambda}_{T}^{-1}\left(\boldsymbol{\Lambda}_{0} \boldsymbol{\mu}_{0}+\mathbf{X}^{\mathrm{T}} \boldsymbol{y}\right), \\ 
a_{T} &=a_{0}+\frac{T}{2}, \quad
b_{T} =b_{0}+\frac{1}{2}\left(\mathbf{y}^{\mathrm{T}} \mathbf{y}+\boldsymbol{\mu}_{0}^{\mathrm{T}} \boldsymbol{\Lambda}_{0} \boldsymbol{\mu}_{0}-\boldsymbol{\mu}_{T}^{\mathrm{T}} \boldsymbol{\Lambda}_{T} \boldsymbol{\mu}_{T}\right). \end{aligned}
$$
Thus, a Gibbs sampler update of $\theta$ only depends on these sufficient statistics. The choice of hyperparameters could be informed by previous empirical studies, or set to resemble a flat or weakly informative prior.

\subsection{Particle Gibbs for the SVM with tractable likelihoods} \label{sec:csmc}

In the SVM context, a particle Gibbs algorithm alternates between sampling from $p(\theta \mid h_{0:T}, r_{1:T})$ and $p(h_{0:T} \mid \theta, r_{1:T})$. As described in the Introduction, sampling from $p(h_{0:T} \mid \theta, r_{1:T})$ requires the use of conditional SMC algorithms. 
The basic idea of cSMC is to take an input trajectory (i.e., a sample for $h_{0:T}$) as a reference and produce an output trajectory via SMC-style propagation; furthermore, cSMC accounts for all random variables generated during propagation via an extended target distribution of higher dimension \citep{chopin2015particle}. The steps for a sweep of particle Gibbs then consist of (i) updating the parameters based on the input trajectory, (ii) generating new trajectories based on the updated parameters and input trajectory, and (iii) selecting an output trajectory as input for the next iteration
\citep{andrieu2010particle}. The use of cSMC in particle Gibbs guarantees the target distribution is admitted as the invariant density.
Here, we briefly review two existing cSMC algorithms that are applicable within particle Gibbs when the likelihood $l_t(r_t \mid h_t)$ can be computed.

The first is the conditional bootstrap filter \citep[cBF,][]{andrieu2010particle} as summarized in Algorithm \ref{CBF}. A key feature is that it preserves the input trajectory $h_{0:T}^*$ throughout propagation and resampling; holding $h_{0:T}^*$  intact, $N-1$ new trajectories are generated ``conditional on'' $h_{0:T}^*$ in SMC fashion; finally one of the $N$ trajectories (i.e., among the input trajectory and the $N-1$ generated trajectories) is selected to be the new input trajectory for the next iteration. As shown in Algorithm \ref{CBF}, for any particle $h_t^{(n)}$ with $n \in \{1,\dots,N\}$ simulated at step $t$, we denote the index of the ancestor particle of $h_t^{(n)}$ by $a_{t-1}^{(n)}$, i.e., $h_t^{(n)}$ is propagated from $h_{t-1}^{(a_{t-1}^{(n)})}$. For simplicity, let $A_{t, t}^{(n)} = n$, $A_{t-1, t}^{(n)} = a_{t-1}^{(n)}$ and $A_{t-l,t}^{(n)} = a_{t-l}^{(A_{t-l+1,t}^{(n)})}$ for $2 \leq l \leq t$; then the $n$-th trajectory can be written as $h_{0:T}^{(n)} = (h_0^{(A_{0,T}^{(n)})}, h_1^{(A_{1,T}^{(n)})}, \dots, h_T^{(A_{T, T}^{(n)})})$ after $t=T$ propagation and resampling steps.

\begin{algorithm}[t]
 \setstretch{1}
 \SetAlgoLined
 input: observations $r_{1:T}$, input trajectory $h^*_{0:T}$, transition density $g_t$, likelihood $l_t$\;
 draw $h_0^{(n)} \sim g_0(h_0)$ for $n=1,\dots,N-1$ and set $h_0^{(N)} = h_0^*$\;
 $w_0^{(n)} = \frac{1}{N}$ for $n=1,\dots,N$\;
 
 \For{t in 1:T}{
 draw index $a_{t-1}^{(n)}$ from $\{(n, w_{t-1}^{(n)})\}_{n=1}^N$ for $n=1,\dots,N-1$\;
 
 set index $a_{t-1}^{(N)} = N$\;
 
 draw $h_t^{(n)} \sim g_t(h_{t} \mid h_{t-1}^{(a_{t-1}^{(n)})})$ for $n=1,\dots,N-1$ and set $h_t^{(N)} = h_t^*$\;
 
 $w_t^{(n)} = l_t(r_t \mid h_t^{(n)})$ for $n=1,\dots,N$\;
 }
 draw index $b$ from $\{(n, w_{T}^{(n)})\}_{n=1}^N$\;
 \Return{$h_{0:T}^{(b)}$}\;
 \caption{Conditional Bootstrap Filter}
 \label{CBF}
\end{algorithm}

The second is the conditional bootstrap filter with ancestor sampling \citep[cBFAS,][]{lindsten2014particle} as summarized in Algorithm \ref{CBFAS}. The cBF keeps the input trajectory $h_{0:T}^*$ intact, i.e., $A_{t, T}^{(N)} = N$, and preserved throughout for all $t\in\{0, \dots, T-1\}$, so the early parts of generated trajectories may tend to closely resemble (or be identical to) the input trajectory, which when too extreme is known as path degeneracy. The key idea of cBFAS is to stochastically perturb the input trajectory by breaking it into pieces via ancestor sampling. Ancestor sampling, as presented in Algorithm \ref{CBFAS}, takes $A_{t, T}^{(N)}$ to be stochastic, so the input trajectory can be partially replaced by other generated trajectories at each step $t$. Consequently, the input trajectory interacts much more with the other trajectories \citep{svensson2015nonlinear}, which can help cBFAS mitigate path degeneracy while maintaining the target distribution as invariant \citep{lindsten2014particle,svensson2015nonlinear}.

\begin{algorithm}[t]
\setstretch{1}
 \SetAlgoLined
 input: observations $r_{1:T}$, input trajectory $h^*_{0:T}$, transition density $g_t$, likelihood $l_t$\;
 draw $h_0^{(n)} \sim g_0(h_0)$ for $n=1,\dots,N-1$ and set $h_0^{(N)} = h_0^*$\;
 $w_0^{(n)} = \frac{1}{N}$ for $n=1,\dots,N$\;
 
 \For{t in 1:T}{
 draw index $a_{t-1}^{(n)}$ from $\{(n, w_{t-1}^{(n)})\}_{n=1}^N$ for $n=1,\dots,N-1$\;
 
 draw index $a_{t-1}^{(N)}$ from $\{(n, w_{t-1}^{(n)}g_{t}(h_t^* \mid h_{t-1}^{(n)}))\}_{n=1}^{N}$\;
 
 draw $h_t^{(n)} \sim g_t(h_{t} \mid h_{t-1}^{(a_{t-1}^{(n)})})$ for $n=1,\dots,N-1$ and set $h_t^{(N)} = h_t^*$\;
 
 $w_t^{(n)} = l_t(r_t \mid h_t^{(n)})$ for $n=1,\dots,N$\;
 }
 draw index $b$ from $\{(n, w_{T}^{(n)})\}_{n=1}^N$\;
 \Return{$h_{0:T}^{(b)}$}\;
 \caption{Conditional Bootstrap Filter with Ancestor Sampling}
 \label{CBFAS}
\end{algorithm}

\subsection{Review of ABC methods and ABC-SMC methods} \label{sec:abc}

To briefly review ABC methods, consider observations $r$ and parameters $\theta$ where the likelihood $l(r \mid \theta)$ is intractable and expensive to evaluate. Consequently, the posterior $p(\theta \mid r) \propto p(\theta)l(r \mid \theta)$ is also intractable. In this setting, ABC can be used for inference of $\theta$ if sampling from $l(\cdot \mid \theta)$ is straightforward. The basic idea of ABC methods is to construct an approximation to the posterior with the help of auxiliary observations (denoted by $u$) and an ABC kernel (denoted by $K_\epsilon(r \mid u)$). The simplest example of an ABC method is a likelihood-free rejection sampling algorithm with a uniform kernel (i.e., $K_\epsilon(r \mid u) \propto \mathbbm{1}_{u-\epsilon < r < u+\epsilon}$ for a chosen $\epsilon>0$) that implements the following steps: (1) sample a candidate $\theta^*$ from the prior $p(\theta)$; (2) sample a realization $u$ from $l(u \mid \theta^*)$; (3) accept $(\theta^*, u)$ if $|r-u| < \epsilon$. Then accepted samples of $(\theta^*, u)$ follow the density defined by
$p_\epsilon(\theta, u \mid r) \propto p(\theta) l(u \mid \theta)K_\epsilon(r\mid u)$, 
which can be viewed as the ABC approximation of the extended distribution $p(\theta, u \mid r)$. In practice, Gaussian kernels (i.e., $K_\epsilon(r \mid u) \propto \exp{[(r-u)^2/(2\epsilon^2})]$) are more commonly used than uniform ones. It is clear that $p_\epsilon(\theta \mid r) \xrightarrow{\text{d}} p(\theta \mid r)$ as $\epsilon \rightarrow 0$ and $p_\epsilon(\theta \mid r) \xrightarrow{\text{d}} p(\theta)$ as $\epsilon \rightarrow \infty$. Therefore, an ABC kernel with a small $\epsilon$ can lead to many near-zero weights of generated candidates $(\theta^*, u)$ (or more rejected samples), while a larger $\epsilon$ can lead to more uniform weights (or more accepted samples); however, the accuracy of the ABC approximation will decrease as a tradeoff.

For an SSM with intractable likelihoods, we can similarly construct the ABC approximation $p_\epsilon(h_{0:T}, u_{1:T} \mid r_{1:T}, \theta)$ of the corresponding extended distribution and sample from it using an SMC algorithm; this is known as ABC-SMC \citep{peters2012sequential}. Thus, the target distribution of ABC-SMC involves the auxiliary observations $u_{1:T}$ for evaluating importance weights and has the form
\begin{align}
    p_\epsilon(h_{0:T}, u_{1:T} \mid r_{1:T}, \theta) \propto g_0(h_0 \mid \theta)\prod_{t=1}^T g_t(h_t \mid h_{t-1}, \theta)  K_\epsilon(r_{t} \mid u_{t}) l_t(u_{t}\mid h_{t}).
    \label{p_eps_abcsmc}
\end{align}
To sample from \eqref{p_eps_abcsmc}, \cite{vankov2019filtering} proposed an ABC-based APF, which is summarized in Algorithm \ref{ABC-APF} and applicable within a particle Metropolis-Hastings algorithm. Within a particle Gibbs algorithm, however, Algorithm \ref{ABC-APF} cannot be directly used (recall that particle Gibbs requires a cSMC setup to admit the target distribution as invariant). Thus, in the following we propose a new ABC-based APF, which we call the ABC-based conditional auxiliary particle filter (ABC-cAPF), that can be embedded within a particle Gibbs algorithm.

\begin{algorithm}[t]
\setstretch{1}
 \SetAlgoLined
 input: observations $r_{1:T}$, input trajectory $h^*_{0:T}$, transition density $g_t$, likelihood $l_t$\;
 draw $h_0^{(n)} \sim g_0(h_0)$ for $n=1,\dots,N-1$ and set $h_0^{(N)} = h_0^*$\;
 $w_0^{(n)} = \frac{1}{N}$ for $n=1,\dots,N$\;
 
 \For{t in 1:T}{

    \For{n in 1:N}{
        $\widetilde{w}_{t-1}^{(n)} = w_{t-1}^{(n)}\int\int K_\epsilon(r_t \mid u_t)l_t(u_t \mid h_t)g_t(h_t \mid h_{t-1}^{(n)}) dh_tdu_t$\;
 
        draw index $a_{t-1}^{(n)}$ from $\{(n,\widetilde{w}_{t-1}^{(n)})\}_{n=1}^N$\;
 
        draw $h_t^{(n)} \sim g_t(h_{t} \mid h_{t-1}^{(a_{t-1}^{(n)})})$\;

        draw $u_t^{(n)} \sim l_t(u_t \mid h_t^{(n)})$\;
 
        $w_t^{(n)} = \frac{w_{t-1}^{(a_{t-1}^{(n)})}}{\widetilde{w}_{t-1}^{(a_{t-1}^{(n)})}}K_\epsilon(r_t \mid u_t^{(n)}) $\;
    
    }
 }
 draw index $b$ from $\{(n, w_{T}^{(n)})\}_{n=1}^N$\;
 \Return{$h_{0:T}^{(b)}$}\;
 \caption{ABC-based Auxiliary Particle Filter}
 \label{ABC-APF}
\end{algorithm}

\subsection{Likelihood-free ABC-based cSMC for the SVM}\label{sec:ABC-cAPF}

Now we consider the SVM with stable distribution as presented in Section \ref{sec:setup}. In the ABC setting with a chosen kernel $K_\epsilon$, Model \eqref{SVM(1)} can then be rewritten as
\begin{align}
\log(h_t) = \tau + \phi \log(h_{t-1}) + \sigma_h \epsilon_t, \quad
    u_t = \sqrt{h_t} \times Z_t, \quad
    r_t \sim K_\epsilon(\cdot \mid u_t)
    \label{ABC-SVM(1)}
\end{align}
with $\epsilon_t \sim N(0, 1)$ and $Z_t \sim SD(\alpha, \beta, \gamma, \delta)$ all independent for $t=1\ldots, T$,  and $h_0 \sim \text{Log-normal}(\tau/(1-\phi), \sigma_h^2/(1-\phi^2))$. Here, each $h_t$, $t=1\ldots, T$ corresponds to an auxiliary observation $u_t$ for importance weight calculation according to  $K_\epsilon$. The extended posterior distribution of interest is then
\begin{align}
p_\epsilon(\theta, h_{0:T}, u_{1:T} \mid r_{1:T}) \propto p(\theta)g_0(h_0 \mid \theta)\prod_{t=1}^T g_t(h_t \mid h_{t-1}, \theta)  K_\epsilon(y_{t} \mid u_{t}) l_t(u_{t}\mid h_{t}),
    \label{p_eps_full}
\end{align}
and the goal is to sample from \eqref{p_eps_full} in a likelihood-free manner, i.e., without computing $l_t$. We shall focus on the particle Gibbs case, and develop ABC-based algorithms to sample from \eqref{p_eps_full} that alternate between sampling from $p(\theta \mid h_{0:T}, u_{1:T}, r_{1:T}) = p(\theta \mid h_{0:T})$ and $p_\epsilon(h_{0:T}, u_{1:T} \mid r_{1:T}, \theta)$.

The cBF and cBFAS algorithms can be applied in the ABC setting with slight modifications: in the step where we draw each $h_t^{(n)} \sim g_t(h_t \mid h_{t-1}^{(a_{t-1}^{(n)})})$, we also draw an auxiliary observation $u_t^{(n)}$ from $l_t(u_t \mid h_t^{(n)})$ and then assign the importance weight $w_t^{(n)} = K_\epsilon(r_t \mid u_t^{(n)})$ for the particle $(h_t^{(n)}, u_t^{(n)})$. We shall call these algorithms the ABC-based conditional bootstrap filter (ABC-cBF) and ABC-based conditional bootstrap filter with ancestor sampling (ABC-cBFAS), respectively. However, these two algorithms can encounter severe weight degeneracy in practice, if we choose a small $\epsilon$ in the ABC kernel to obtain an accurate approximation of the true target distribution. Thus, in the following we also propose a novel ABC-based cSMC algorithm using the auxiliary particle filter as the building block.

We shall call this third algorithm the ABC-based conditional auxiliary particle filter (ABC-cAPF), as presented in Algorithm \ref{ABC-CAPF}. To initialize the ABC-cAPF at $t=0$, $N-1$ particles, $\{h_0^{(n)}\}_{n=1}^{N-1}$, are sampled from the log-normal distribution and $h_0^{(N)}$ is assigned the initial log-volatility of the input trajectory. Then after the $(t-1)$-th propagation step ($t = 1, \ldots, T$), the $N$-th particle is set to be the input trajectory $h_{0:t-1}^*$ and the remaining $N-1$ particles $\{(h_{0:t-1}^{(n)}, u_{1:t-1}^{(n)})\}_{n=1}^{N-1}$ are generated with weights that satisfy
$$
w_{t-1}^{(n)} \propto g_0(h_0^{(A_{0,t-1}^{(n)})}) \prod_{s=1}^{t-1} K_\epsilon(r_s \mid u_s^{(A_{s,t-1}^{(n)})})l_s(u_s^{(A_{s,t-1}^{(n)})} \mid h_s^{(A_{s,t-1}^{(n)})})g_s(h_s^{(A_{s,t-1}^{(n)})} \mid h_{s-1}^{(A_{s-1,t-1}^{(n)})})
$$
for $n=1,\dots,N-1$. Following the concept of properly weighted particles in SMC \citep{liu2001monte}, this weight ensures that the set of $N-1$ generated particles is properly weighted with respect to $p_\epsilon(h_{0:t-1}, u_{1:t-1} \mid r_{1:t-1}, \theta)$. As an important feature of the APF \citep{pitt1999filtering} when propagating the particles from $t-1$ to $t$, the tempered weights $\widetilde{w}_{t-1}^{(n)}$, i.e., the adjusted importance weights that incorporate the one-step-ahead observation $r_t$, are computed and utilized; in the ABC setting, the computation of these tempered weights require an intractable integration, namely
\begin{equation}
    \widetilde{w}_{t-1}^{(n)} =  w_{t-1}^{(n)}p(r_t \mid h_{t-1}) = w_{t-1}^{(n)} \int\int K_\epsilon(r_t \mid u_t) l_t(u_t \mid h_t) g_t(h_t \mid h_{t-1}^{(n)}) dh_tdu_t
    \label{eq:tempering}
\end{equation}
where the double integral does not have a closed form; \cite{vankov2019filtering} suggest approximating it with a more heavy-tailed distribution such as a $t$-class distribution. 

Specifically, given $r_t = \sqrt{h_t} \times Z_{\text{stable}}$, we construct an approximation of the stable distribution noise with a standard Cauchy noise as $\log(r_t^2) = \log(h_t) + \log(Z_{\text{Cauchy}}^2)$ (thus $E\{\log(r_t^2) \mid h_{t-1}\} = \tau + \phi\log(h_{t-1})$ and $Var\{\log(r_t^2) \mid h_{t-1}\} = \sigma_h^2 + \pi^2$). Then we approximate $\log(r_t^2)$ with a linear combination of $\log(Z_{\text{Cauchy}}^2)$ and $h_{t-1}$ to match the conditional mean and variance, which gives $\log(Z_{\text{cauchy}}^2) = \sqrt{\frac{\pi^2}{\sigma_h^2 + \pi^2}} \{ \log(r_t^2) - \tau - \phi\log(h_{t-1}) \}$, so $p(r_t \mid h_{t-1})$ is approximately proportional to
$$
\left(1 + (r_t^{2})^{\sqrt{\frac{\pi^2}{\sigma_h^2 + \pi^2}}}\exp\left[-\sqrt{\frac{\pi^2}{\sigma_h^2 + \pi^2}}\left\{\tau+\phi\log(h_{t-1})\right\}\right]\right)^{-1}.
$$
Note that the quality of the integral approximation in \eqref{eq:tempering} does not influence the validity of ABC-cAPF; however, the tempered weights should cover the high-density regions of the true importance weights for a more efficient algorithm.

For $t = 1,\dots,T$, resampling and propagation are implemented with the tempered weights, but otherwise similar to the cBF: (i) resample the first $N-1$ particles from $\{(h_{0:t-1}^{(n)}, u_{1:t-1}^{(n)})\}_{n=1}^N$ proportional to the tempered weights $\{\widetilde{w}_{t-1}^{(n)}\}_{n=1}^N$ while leaving the $N$-th particle intact as the input trajectory; (ii) propagate the resampled particles as in the cBF; (iii) compute the importance weights of the propagated particles so that cAPF targets the same density as cBF, i.e., the weights of the propagated particles $\{(h_{0:t}^{(n)}, u_{1:t}^{(n)})\}_{n=1}^{N-1}$ satisfy
$$
w_{t}^{(n)} \propto g_0(h_0^{(A_{0,t}^{(n)})}) \prod_{s=1}^{t} K_\epsilon(r_s \mid u_s^{(A_{s,t}^{(n)})})l_s(u_s^{(A_{s,t}^{(n)})} \mid h_s^{(A_{s,t}^{(n)})})g_s(h_s^{(A_{s,t}^{(n)})} \mid h_{s-1}^{(A_{s-1,t}^{(n)})}).
$$

\begin{algorithm}[t]
\setstretch{1}
 \SetAlgoLined
 input: observations $r_{1:T}$, input trajectory $h^*_{0:T}$, transition density $g_t$, likelihood $l_t$\;
 draw $h_0^{(n)} \sim g_0(h_0)$ for $n=1,\dots,N-1$ and set $h_0^{(N)} = h_0^*$\;
 $w_0^{(n)} = \frac{1}{N}$ for $n=1,\dots,N$\;
 
 \For{t in 1:T}{

 compute $\widetilde{w}_{t-1}^{(n)}$ using $w_{t-1}^{(n)}$, $r_t$ and $h_{t-1}^{(n)}$ for $n=1,\dots,N$\;
 
 draw index $a_{t-1}^{(n)}$ from $\{(n,\widetilde{w}_{t-1}^{(n)})\}_{n=1}^N$ for $n=1,\dots,N-1$ and set index $a_{t-1}^{(N)} = N$\;
 
 draw $h_t^{(n)} \sim g_t(h_{t} \mid h_{t-1}^{(a_{t-1}^{(n)})})$ for $n=1,\dots,N-1$ and set $h_t^{(N)} = h_t^*$\;

 draw $u_t^{(n)} \sim l_t(u_t \mid h_t^{(n)})$ for $n=1,\dots,N$\;
 
 $w_t^{(n)} = \frac{w_{t-1}^{(a_{t-1}^{(n)})}}{\widetilde{w}_{t-1}^{(a_{t-1}^{(n)})}}K_\epsilon(r_t \mid u_t^{(n)})$ for $n=1,\dots,N-1$ and $w_t^{(N)} = K_\epsilon(r_t \mid u_t^{(n)})$\;
 }
 draw index $b$ from $\{(n, w_{T}^{(n)})\}_{n=1}^N$\;
 \Return{$h_{0:T}^{(b)}$}\;
 \caption{ABC-based Conditional Auxiliary Particle Filter}
 \label{ABC-CAPF}
\end{algorithm}

\begin{algorithm}[t]
\setstretch{1}
 \SetAlgoLined
 initialization: $\theta[0]$, input trajectory $h^*_{0:T}[0]$, burning size $S$, sample size $N$
 
 \For{t in 1:($S+N$)}{

 run Algorithm \ref{ABC-CAPF} with $\theta[t-1]$ and $h^*_{0:T}[t-1]$ and output $h^*_{0:T}[t]$\;

 update the sufficient statistics with $h^*_{0:T}[t]$\;

 sample $\theta[t]$ from the posterior $NIG(a_T, b_T, \boldsymbol{\mu}_T, \boldsymbol{\Lambda}_T)$\;
 }

 \Return{$\theta[(S+1):N]$}\;
 \caption{ABC-based Particle Gibbs with the Conditional Auxiliary Particle Filter}
 \label{ABC-PG-CAPF}
\end{algorithm}

Each of the three algorithms, i.e., ABC-cBF, ABC-cBFAS, and ABC-cAPF, can be combined with a Gibbs sampler for the parameters $p(\theta \mid h_{0:T}, u_{1:T}, r_{1:T})$ to construct particle Gibbs samplers for \eqref{p_eps_full}. We call these algorithms the ABC-based particle Gibbs with ABC-cBF (ABC-PG-cBF), ABC-based particle Gibbs with ABC-cBFAS (ABC-PG-cBFAS), and ABC-based particle Gibbs with the conditional auxiliary particle filter (ABC-PG-cAPF), respectively. Algorithm \ref{ABC-PG-CAPF} shows how the ABC-cAPF is embedded within the particle Gibbs sampler: each Gibbs sweep uses the ABC-cAPF to draw an output trajectory for $h_{0:T}$ (given the current draw of $\theta$ and the input trajectory), and then draws an updated $\theta$ from its closed-form NIG conditional posterior (given the output trajectory, which becomes the input trajectory for the next iteration). The ABC-PG-cBF and ABC-PG-cBFAS algorithms are constructed analogously. To initialize the particle Gibbs sampler, we draw parameters from $p(\theta)$ and an input trajectory $h_{0:T}^*$ from the log-normal distribution given those parameters.
\begin{proposition} \label{thm:proper}
The proposed particle Gibbs algorithm (ABC-PG-cAPF in Algorithm \ref{ABC-PG-CAPF}) admits the target distribution, i.e., the posterior distribution $p_\epsilon(h_{0:T}, u_{1:T}, \theta \mid r_{1:T})$ in \eqref{p_eps_full}, as the invariant density under some mild assumptions.
\end{proposition}
\begin{proof}
See Section S1 of the Supplementary Material.
\end{proof}
While ABC-PG-cAPF is proposed with Model \eqref{ABC-SVM(1)} as the focus in this paper, the computational strategy is not limited to this specific problem and can be applied to other SSMs with intractable likelihoods, e.g., stochastic kinetic models \citep{owen2015scalable,lowe2023accelerating} and models with likelihoods that follow g-and-k distributions \citep{rayner2002numerical,drovandi2011likelihood}.

\section{Simulation study} \label{sec:simulation}

We consider the model described in Equation \eqref{ABC-SVM(1)} with a second-order stationary log-volatility process. Since $h_t$ conditional on $h_{t-1}$ follows a log-normal distribution, we have $E(h_t) = \exp\left\{\frac{\gamma}{1-\phi} + \frac{\sigma_h^2}{2(1-\phi^2)}\right\}$ and $Var(h_t)=E^2(h_t)\left\{\exp\left(\frac{\sigma_h^2}{1-\phi^2}\right)-1\right\}$. Therefore, the (squared) coefficient of variation, $\mathrm{CV} = Var(h_t)/E^2(h_t)$, can be written as $\mathrm{CV} = \exp\left(\frac{\sigma_h^2}{1-\phi^2}\right)-1$. 
Practical ranges of $\phi$ that have been suggested from empirical analyses are $\phi \in [0.9, 0.98]$, or more loosely, $\phi \in [0.8, 0.995]$ \citep{jacquier1994bayesian}; following their simulation study we take $\phi \in \{0.9, 0.95, 0.98\}$, $\mathrm{CV} \in \{0.1, 1, 10\}$, and set $E(h_t) = 0.0009$. Given the values of $\phi$, $\mathrm{CV}$, and $E(h_t)$, the corresponding values of $\gamma$ and $\sigma_h^2$ can be easily computed. Following the setup for $SD(\alpha, \beta, \gamma, \delta)$ in \cite{vankov2019filtering}, we hold $\gamma = 1$ and $\delta = 0$ fixed throughout and consider $(\alpha, \beta) = (1.75, 0.1), (1.7, 0.3), (1.5, -0.3)$ respectively as three experimental settings; (1.75, 0.1) gives the least heavy-tailed stable distribution and the least (right) skewness and is expected to be the easiest to handle; (1.5, -0.3) gives the most heavy-tailed stable distribution and the most (left) skewness and is expected to be the hardest to handle; (1.7, 0.3) has heavy-tailedness between that of the above two cases and the most (right) skewness. Following \cite{yang2018sequential}, we assign a weakly informative $NIG$ conjugate prior for $\theta$ with $a_0=2$, $b_0=0.5$, $\boldsymbol{\mu}_0 = \begin{bmatrix}
0\\
0.9
\end{bmatrix}$ and $\boldsymbol{\Lambda}_0 = \begin{bmatrix}
1 & 0\\
0 & 1
\end{bmatrix}$.

For each combination of $\phi$, CV, and $(\alpha, \beta)$, we simulated a time series of length $T=100$ observations. We then ran each of the particle Gibbs samplers (ABC-PG-cBF, ABC-PG-cBFAS, and ABC-PG-cAPF) using $N=100$ particles for the embedded cSMC algorithm and a Gaussian kernel with $\epsilon = 0.001$. We discarded the first 2000 iterations as burn-in and took the 5000 iterations after the burn-in period as the posterior sample. The posterior means are treated as the Bayes' estimates for the parameters. To compare the performances of the different PG samplers, we repeated this procedure for 100 simulated datasets. The simulation results, summarized by computing the root-mean-squared errors (RMSEs) of the parameter estimates over the 100 repetitions, are presented in Tables \ref{Table:simu_1}, \ref{Table:simu_2}, and \ref{Table:simu_3} for $(\alpha, \beta) =  (1.75, 0.1), (1.7, 0.3), (1.5, -0.3)$, respectively.

\begin{table}[h]
    \centering
    \renewcommand{\arraystretch}{1}
    \begin{tabular}{cc|cccc}
    \hline
    CV & $\phi$ & Algorithm & RMSE $\tau$ & RMSE $\phi$ & RMSE $\sigma_h^2$ \\ 

        \hline
       ~ & ~ & ABC-PG-cBF & 0.518 & 0.068 & 0.463 \\
        10 & 0.9 & ABC-PG-cBFAS & 0.594 & 0.078 & 0.505 \\
        ~ & ~ & ABC-PG-cAPF & 0.173 & 0.025 & 0.141 \\
        \hline
    
        \hline
       ~ & ~ & ABC-PG-cBF & 0.647 & 0.086 & 0.420  \\
        10 & 0.95 & ABC-PG-cBFAS & 0.747 & 0.099 & 0.489  \\
        ~ & ~ & ABC-PG-cAPF & 0.324 & 0.049 & 0.145  \\
        \hline

        \hline
        ~ & ~ & ABC-PG-cBF & 0.817 & 0.112 & 0.414 \\
        10 & 0.98 & ABC-PG-cBFAS & 0.916 & 0.131 & 0.510 \\
        ~ & ~ & ABC-PG-cAPF & 0.547 & 0.079 & 0.214 \\
        \hline
        
        \hline
        ~ & ~ & ABC-PG-cBF & 0.571 & 0.080 & 0.441  \\
        1 & 0.9 & ABC-PG-cBFAS & 0.663 & 0.094 & 0.486 \\
        ~ & ~ & ABC-PG-cAPF & 0.157 & 0.029 & 0.179 \\
        \hline

        \hline
        ~ & ~ & ABC-PG-cBF & 0.819 & 0.116 & 0.466 \\
        1 & 0.95 & ABC-PG-cBFAS & 0.946 & 0.134 & 0.524 \\
        ~ & ~ & ABC-PG-cAPF & 0.425 & 0.068 & 0.216 \\
        \hline

        \hline
        ~ & ~ & ABC-PG-cBF & 0.988 & 0.142 & 0.449 \\
        1 & 0.98 & ABC-PG-cBFAS & 1.125 & 0.161 & 0.516 \\
        ~ & ~ & ABC-PG-cAPF & 0.645 & 0.099 & 0.245 \\

        \hline
        ~ & ~ & ABC-PG-cBF & 0.578 & 0.084 & 0.455  \\
        0.1 & 0.9 & ABC-PG-cBFAS & 0.709 & 0.103 & 0.520 \\
        ~ & ~ & ABC-PG-cAPF & 0.171 & 0.033 & 0.254 \\
        \hline

        \hline
        ~ & ~ & ABC-PG-cBF & 0.855 & 0.123 & 0.463  \\
        0.1 & 0.95 & ABC-PG-cBFAS & 1.012 & 0.146 & 0.526  \\
        ~ & ~ & ABC-PG-cAPF & 0.484 & 0.078 & 0.259  \\
        \hline

        \hline
        ~ & ~ & ABC-PG-cBF & 1.038 & 0.150 & 0.459  \\
        0.1 & 0.98 & ABC-PG-cBFAS & 1.214 & 0.176 & 0.534  \\
        ~ & ~ & ABC-PG-cAPF & 0.699 & 0.109 & 0.268 \\
        \hline

    \end{tabular}
    
    \caption{RMSEs of the SVM parameter estimates for the three particle Gibbs algorithms, based on 100 simulated datasets with $T=100$, $N=100$, $Z_t \sim SD(1.75, 0.1, 1, 0)$, $\epsilon = 0.001$.}
    \label{Table:simu_1}
\end{table}

\begin{table}[h]
    \centering
    \renewcommand{\arraystretch}{1}
    \begin{tabular}{cc|cccc}
    \hline
    CV & $\phi$ & Algorithm & RMSE $\tau$ & RMSE $\phi$ & RMSE $\sigma_h^2$ \\ 

        \hline
       ~ & ~ & ABC-PG-cBF & 0.514 & 0.067 & 0.455  \\
        10 & 0.9 & ABC-PG-cBFAS & 0.610 & 0.079 & 0.503 \\
        ~ & ~ & ABC-PG-cAPF & 0.170 & 0.025 & 0.144 \\
        \hline
    
        \hline
       ~ & ~ & ABC-PG-cBF & 0.653 & 0.087 & 0.431 \\
        10 & 0.95 & ABC-PG-cBFAS & 0.768 & 0.102 & 0.502 \\
        ~ & ~ & ABC-PG-cAPF & 0.335 & 0.051 & 0.156 \\
        \hline

        \hline
        ~ & ~ & ABC-PG-cBF & 0.879 & 0.120 & 0.475 \\
        10 & 0.98 & ABC-PG-cBFAS & 0.993 & 0.136 & 0.530 \\
        ~ & ~ & ABC-PG-cAPF & 0.557 & 0.080 & 0.218 \\
        \hline

        \hline
        ~ & ~ & ABC-PG-cBF & 0.564 & 0.080 & 0.435 \\
        1 & 0.9 & ABC-PG-cBFAS & 0.686 & 0.098 & 0.503 \\
        ~ & ~ & ABC-PG-cAPF & 0.148 & 0.029 & 0.178 \\
        \hline

        \hline
        ~ & ~ & ABC-PG-cBF & 0.831 & 0.117 & 0.477 \\
        1 & 0.95 & ABC-PG-cBFAS & 0.959 & 0.136 & 0.531 \\
        ~ & ~ & ABC-PG-cAPF & 0.443 & 0.071 & 0.226 \\
        \hline

        \hline
        ~ & ~ & ABC-PG-cBF & 1.025 & 0.148 & 0.472 \\
        1 & 0.98 & ABC-PG-cBFAS & 1.167 & 0.167 & 0.550 \\
        ~ & ~ & ABC-PG-cAPF & 0.677 & 0.104 & 0.265 \\

        \hline
        ~ & ~ & ABC-PG-cBF & 0.622 & 0.089 & 0.495 \\
        0.1 & 0.9 & ABC-PG-cBFAS & 0.745 & 0.109 & 0.548 \\
        ~ & ~ & ABC-PG-cAPF & 0.200 & 0.037 & 0.265 \\
        \hline

        \hline
        ~ & ~ & ABC-PG-cBF & 0.892 & 0.129 & 0.487 \\
        0.1 & 0.95 & ABC-PG-cBFAS & 1.041 & 0.151 & 0.550 \\
        ~ & ~ & ABC-PG-cAPF & 0.482 & 0.078 & 0.261 \\
        \hline

        \hline
        ~ & ~ & ABC-PG-cBF & 1.101 & 0.160 & 0.504 \\
        0.1 & 0.98 & ABC-PG-cBFAS & 1.246 & 0.181 & 0.560 \\
        ~ & ~ & ABC-PG-cAPF & 0.707 & 0.110 & 0.274 \\
        \hline

    \end{tabular}
    \caption{RMSEs of the SVM parameter estimates for the three particle Gibbs algorithms, based on 100 simulated datasets with $T=100$, $N=100$, $Z_t \sim SD(1.7, 0.3, 1, 0)$, $\epsilon = 0.001$.}
    \label{Table:simu_2}
\end{table}

\begin{table}[h]
    \centering
    \renewcommand{\arraystretch}{1}
    \begin{tabular}{cc|cccc}
    \hline
    CV & $\phi$ & Algorithm & RMSE $\tau$ & RMSE $\phi$ & RMSE $\sigma_h^2$ \\ 

        \hline
       ~ & ~ & ABC-PG-cBF & 0.681 & 0.088 & 0.614 \\
        10 & 0.9 & ABC-PG-cBFAS & 0.812 & 0.105 & 0.722 \\
        ~ & ~ & ABC-PG-cAPF & 0.174 & 0.027 & 0.158 \\
        \hline
    
        \hline
       ~ & ~ & ABC-PG-cBF & 0.830 & 0.110 & 0.596 \\
        10 & 0.95 & ABC-PG-cBFAS & 0.986 & 0.130 & 0.702 \\
        ~ & ~ & ABC-PG-cAPF & 0.378 & 0.057 & 0.179 \\
        \hline

        \hline
        ~ & ~ & ABC-PG-cBF & 1.007 & 0.137 & 0.566 \\
        10 & 0.98 & ABC-PG-cBFAS & 1.174 & 0.158 & 0.685 \\
        ~ & ~ & ABC-PG-cAPF & 0.615 & 0.090 & 0.253 \\
        \hline

        \hline
        ~ & ~ & ABC-PG-cBF & 0.714 & 0.101 & 0.569 \\
        1 & 0.9 & ABC-PG-cBFAS & 0.889 & 0.126 & 0.671 \\
        ~ & ~ & ABC-PG-cAPF & 0.174 & 0.035 & 0.202 \\
        \hline

        \hline
        ~ & ~ & ABC-PG-cBF & 0.960 & 0.136 & 0.580 \\
        1 & 0.95 & ABC-PG-cBFAS & 1.153 & 0.163 & 0.684 \\
        ~ & ~ & ABC-PG-cAPF & 0.482 & 0.078 & 0.254 \\
        \hline

        \hline
        ~ & ~ & ABC-PG-cBF & 1.145 & 0.164 & 0.579 \\
        1 & 0.98 & ABC-PG-cBFAS & 1.338 & 0.192 & 0.681 \\
        ~ & ~ & ABC-PG-cAPF & 0.708 & 0.111 & 0.290 \\

        \hline
        ~ & ~ & ABC-PG-cBF & 0.739 & 0.108 & 0.587 \\
        0.1 & 0.9 & ABC-PG-cBFAS & 0.896 & 0.131 & 0.674 \\
        ~ & ~ & ABC-PG-cAPF & 0.204 & 0.041 & 0.286 \\
        \hline

        \hline
        ~ & ~ & ABC-PG-cBF & 1.018 & 0.148 & 0.583 \\
        0.1 & 0.95 & ABC-PG-cBFAS & 1.209 & 0.176 & 0.682 \\
        ~ & ~ & ABC-PG-cAPF & 0.553 & 0.091 & 0.306 \\
        \hline

        \hline
        ~ & ~ & ABC-PG-cBF & 1.210 & 0.176 & 0.591 \\
        0.1 & 0.98 & ABC-PG-cBFAS & 1.392 & 0.203 & 0.671 \\
        ~ & ~ & ABC-PG-cAPF & 0.754 & 0.120 & 0.309 \\
        \hline

    \end{tabular}
    \caption{RMSEs of the SVM parameter estimates for the three particle Gibbs algorithms, based on 100 simulated datasets with $T=100$, $N=100$, $Z_t \sim SD(1.5, -0.3, 1, 0)$, $\epsilon = 0.001$.}
    \label{Table:simu_3}
\end{table}

The numerical results show that the proposed ABC-PG-cAPF outperforms the corresponding cBF-based algorithms with parameter RMSEs that are $\sim$30--70\% smaller, depending on the scenario considered. The Gaussian ABC kernel tends to produce very uneven particle weights and the cBF-based algorithms become hindered by weight degeneracy. While ancestor sampling (in cBFAS) is designed to tackle path degeneracy, doing so has limited effectiveness for weight degeneracy and appears to be slightly worse than cBF in these scenarios. In contrast, the weight tempering strategy of the cAPF helps reduce the variability in the weights, which leads to more plausible trajectories being sampled, and in turn improves the accuracy of the parameter estimates.

Some specific patterns in the results can be noted: the RMSEs tend to be higher for a smaller $\sigma_h^2$ (i.e., a larger $\phi$ or a smaller CV) or a more heavy-tailed $Z_t$ (i.e., a smaller $\alpha$). 
Intuitively, volatility is more persistent for a smaller $\sigma_h^2$; however, $r_t$ can still have large fluctuations due to the heavy tails of $Z_t$, so it may be difficult in practice to disentangle the variability coming from $Z_t$ and $h_t$. More heavy-tailed $Z_t$'s make the estimation problem even more challenging. Finally in the ABC setting, a smaller $\epsilon$ is generally preferred for a better ABC approximation of the target distribution; however, the cSMC algorithms may suffer from more severe weight degeneracy when $\epsilon$ is very small. When we decreased $\epsilon$ to 0.0005, the parameter RMSEs were worse for all three algorithms, but ABC-PG-cAPF was the most robust; see Section S2 in the Supplementary Material for details.

\section{Application: the S\&P 500 Index during the Financial Crisis in 2008} \label{section:application}

We now apply the ABC-PG-cAPF algorithm to fit an SVM to time-series data obtained from the Standard \& Poor 500 (S\&P 500) index. Given the time-series of the daily price (i.e., the average of the open and close price on each day), we computed the daily returns for the period January 2008 to March 2009. These are displayed via the black line in the left panel of Figure \ref{figure:SP_return}. The large fluctuations around October 2008 indicate the climax of the well-known global financial crisis. Following the suggestion of \cite{kabavsinskas2009alpha} to take stable distribution parameters in the ranges $\alpha \in (1.65, 1.8)$ and $\beta \in (-0.017, 0.2)$ for financial data, we set $\alpha$ and $\beta$ to be the midpoints of these ranges for our analysis.

\begin{figure}[t]
\begin{center}
\begin{tabular}{cc}
     \includegraphics[scale=0.41]{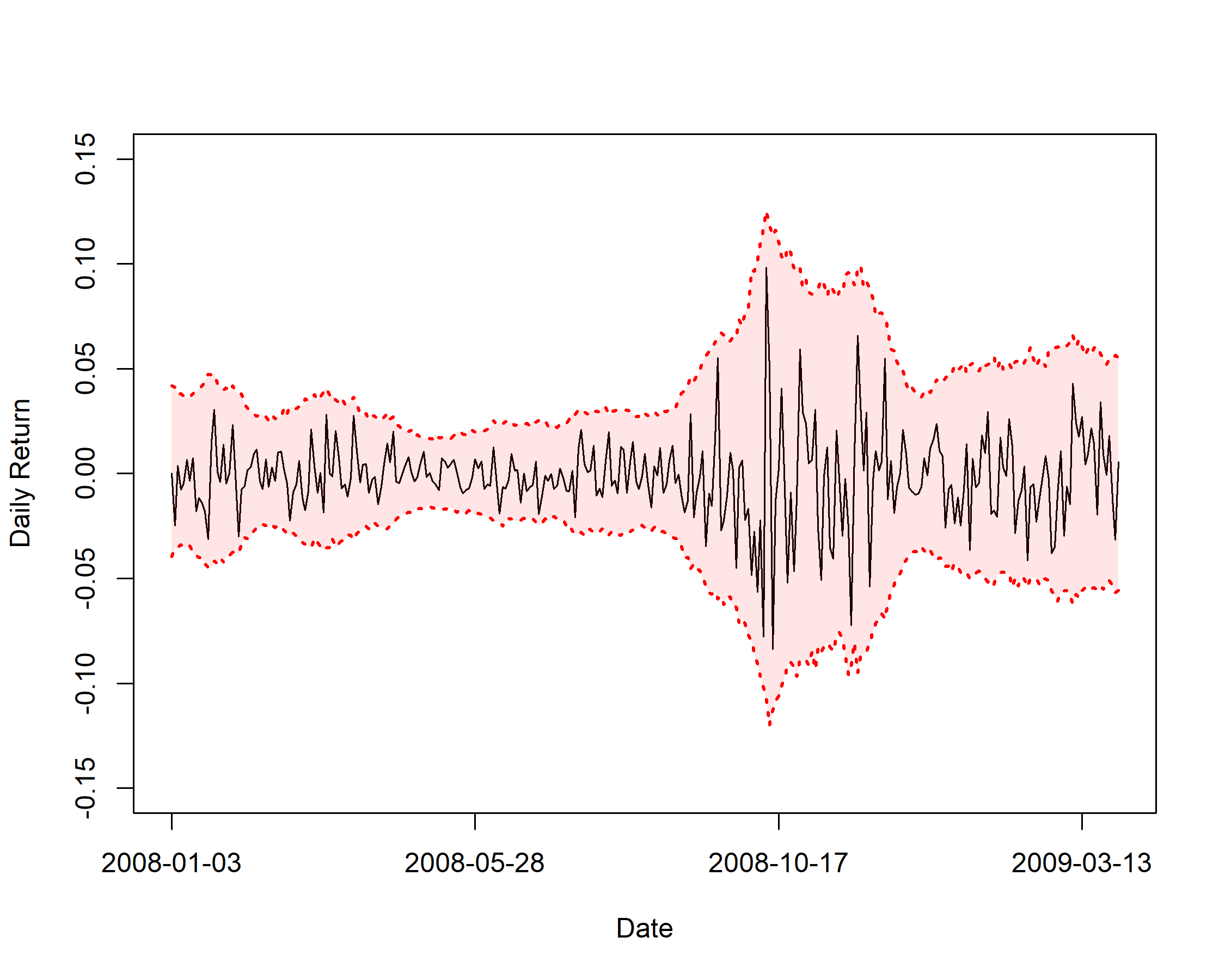} &  
     \includegraphics[scale=0.41]{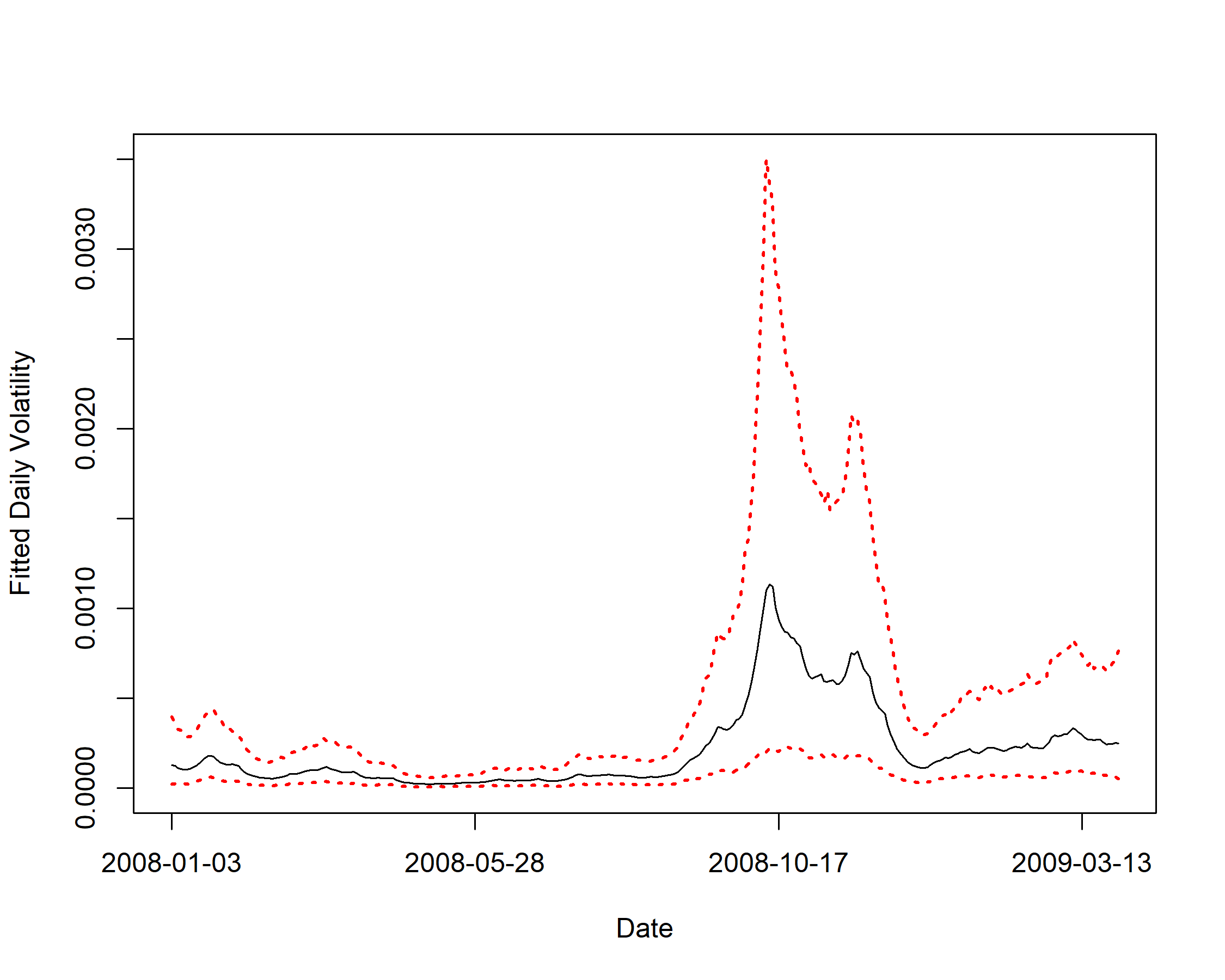}\\
      
\end{tabular}
	
	\caption{The left panel presents the daily returns of the S\&P 500 index from January 2008 to March 2009. The large fluctuations around October 2008 indicate the climax of the global financial crisis. The 2.5\% and 97.5\% quantiles of the sampled returns are shown with red dashed lines and the corresponding 95\% credible interval for the returns is highlighted in red. The right panel presents the fitted daily volatility (black line) with a 95\% credible interval (red dashed lines); the high volatility around the climax is well-captured.}
	\label{figure:SP_return}
\end{center}
\end{figure}

We estimated the SVM parameters based on these data, running 2000 burn-in iterations of ABC-PG-cAPF followed by 5000 sampling iterations, using $N=500$ particles and a Gaussian ABC kernel with $\epsilon = 0.001$. The Bayes' estimates (posterior means) and credible intervals of the parameters are reported in Table \ref{Table:application}. In the right panel of Figure \ref{figure:SP_return}, we plot the fitted daily volatility $\{h_{t}^*\}_{t=1}^T$ along with a central 95\% credible interval based on the posterior samples; the estimated volatility peaks at the climax of the financial crisis. Lastly, the estimated 95\% credible intervals for the returns are also superimposed on the left panel, which were computed via the samples generated from $l(r_{t} \mid h_{t}^*)$ for each draw of $h_{t}^*$ at each time $t = 1, \ldots, T$ and indicate a good fit to the data.

\begin{table}[h]
    \centering
    \renewcommand{\arraystretch}{1}
    \begin{tabular}{ccc}
    \hline
    Parameter & Estimate & 95\% credible interval \\ 

        \hline
       $\tau$ & -0.294 & (-0.639, -0.042) \\
       $\phi$ & 0.967 & (0.930, 0.995) \\
       $\sigma_h^2$ & 0.098 & (0.052, 0.174) \\
        \hline

    \end{tabular}
    \caption{Estimates and credible intervals of $\tau$, $\phi$ and $\sigma_h^2$ based on fitting an SVM to S\&P 500 Index data from January 2008 to March 2009. The numerical results were obtained from the ABC-PG-cAPF algorithm with $N=500$, $Z_t \sim SD(1.725, 0.0915, 1, 0)$, $\epsilon = 0.001$, and 5000 posterior samples.}
    \label{Table:application}
\end{table}

\section{Conclusion and Discussion} \label{section: conclusion}

In this paper, we proposed an ABC-based cAPF embedded within a particle Gibbs sampler for likelihood-free inference of the SVM. Our proposed sampler builds upon a rich SMC and SVM literature, e.g., the idea of using MCMC for parameter estimation in SVMs \citep{jacquier1994bayesian}, particle Gibbs samplers for state space models \citep{andrieu2010particle}, and ABC-based PMCMC for SVM inference with intractable likelihoods \citep{vankov2019filtering}.
Compared to existing particle Gibbs samplers, the proposed ABC-PG-cAPF algorithm produces more accurate parameter estimates with the help of its weight tempering strategy, as demonstrated in the simulation study.

Our sampler can be adapted for broader use with different models and setups. First, if there are any model parameters without closed-form conditional distributions (e.g., $\epsilon_t$ chosen to be a $t$-distribution and thus no conjugacy available for $\sigma_h^2$), or if the hyperparameters $(\alpha, \beta)$ of the stable distribution also need to be estimated, this can be handled by incorporating an additional block of Metropolis updates within the particle Gibbs algorithm. However, the tuning of Metropolis kernels needs to be done carefully \citep{vankov2019filtering}. Second, the computation of the tempered weights for cAPF can be adapted as appropriate to cover the high-density regions of the true importance weights. For example, if the likelihood involves a stable distribution with $0.5 < \alpha < 1$, the Lévy distribution (which corresponds to a stable distribution with $\alpha = 0.5$) can be a better choice of approximating distribution than a Cauchy. However, for a stable distribution with $\alpha < 0.5$ or other kinds of intractable distributions, alternative approximation schemes or numerical methods might be necessary to perform weight tempering efficiently. Finally, other application areas that use SSMs with intractable likelihoods could be considered, e.g., stochastic kinetic models in systems biology \citep{owen2015scalable,lowe2023accelerating}.

\subsection*{Acknowledgments}

This work was partially supported by Discovery Grant RGPIN-2019-04771 from the Natural Sciences and Engineering Research Council of Canada.

\bibliographystyle{apalike}
\bibliography{reference}

\appendix

\newpage
\section*{Supplementary material}

The Supplementary Material contains the proof of Proposition 1 in Section \ref{sec:ABC-cAPF} and the additional experiments described in Section \ref{sec:simulation}.

\renewcommand\thetable{S\arabic{table}} 
\renewcommand{\thesection}{S\arabic{section}}
\renewcommand{\theequation}{S\arabic{equation}}

\section{Proof of Proposition 1}
In the ABC setting, we extend the dimension of the target distribution by introducing a sequence of auxiliary observations $u_{1:T}$. For simplicity, write $x_0 = h_0$ and $x_t = (h_t, u_t)$ for $t = 1,\dots,T$, and we denote the target distribution by $p(x_{0:T}, \theta \mid r_{1:T})$ (dropping $\epsilon$ to simplify notation as it is known and fixed). Then we have the same target distribution as \cite{andrieu2010particle}, where 
they showed that particle Gibbs with the conditional bootstrap filter (cBF), i.e., Algorithm 1 in the main text, admits the target distribution as the invariant distribution under some mild assumptions (i.e., Assumptions 5-7 in \cite{andrieu2010particle}). Here, we shall show that Algorithm 5 in the main text also admits the target distribution as the invariant distribution under the same mild assumptions.

The sweep of particle Gibbs has three steps (see Section 4.5 in \cite{andrieu2010particle}): for cBF, given an input trajectory $x^*_{0:T}$,
\begin{enumerate}
	\item update the parameters, denoted by $\theta^*$
	\item generate new trajectories with ancestor indices conditional on $x^*_{0:T}$ and $\theta^*$ using cBF
	\item draw a trajectory from the union of the generated trajectories and the input trajectory, as the new input trajectory for the next iteration
\end{enumerate}
and we propose to substitute the cBF for Step 2 with the conditional auxiliary particle filter (cAPF). Therefore, to prove the proposition, it is sufficient to show that with the same input trajectory, the cBF and the cAPF target the same distribution.

Write the target distribution of the cBF at each $t$ as $\pi_t(x_{0:t} \mid \theta, r_{1:t})$ for $t = 1,\dots,T$ and $\pi_0 = g_0(x_0 \mid \theta)$ (see Equation (33) in \cite{andrieu2010particle} for the explicit form); for a particle $x_{0:t}$, its particle weight generated by cBF (Algorithm 1) is denoted by $w_t^{(\text{cBF})}(x_{0:t})$, and its particle weight generated by cAPF (Algorithm 4) is denoted by $w_t^{(\text{cAPF})}(x_{0:t})$. Then for the input trajectory $x^*_{0:t}$, it is obvious that we have $w_t^{(\text{cBF})}(x^*_{0:t}) = w_t^{(\text{cAPF})}(x^*_{0:t})$. The input trajectory is set to be the $N$-th particle, i.e., $x^{(N)}_{0:t} = x^*_{0:t}$, for each $t$. Furthermore, for each generated trajectory $x^{(n)}_{0:t}$, $n = 1, \ldots, N-1$, we have 
\begin{align}
E\left\{w_t^{(\text{cAPF})}(x^{(n)}_{0:t})h(x^{(n)}_{0:t})\right\} &= \int w_t^{(\text{cAPF})}(x^{(n)}_{0:t})h(x^{(n)}_{0:t})\eta_t(x^{(n)}_{0:t}) dx^{(n)}_{0:t} \nonumber \\
&= \int \frac{w_t^{(\text{cBF})}(x^{(n)}_{0:t})h(x^{(n)}_{0:t})}{\eta_t(x^{(n)}_{0:t})} \eta_t(x^{(n)}_{0:t}) dx^{(n)}_{0:t} \nonumber 
\\
&= E\left\{w_t^{(\text{cBF})}(x^{(n)}_{0:t})h(x^{(n)}_{0:t})\right\}
\label{eq:proper_weight}
\end{align}
where $h$ is any squared integrable function and $\eta_t(x^{(n)}_{0:t}) = \prod_{s=0}^{t-1} \widetilde{w}_{s}^{(A_{s,t}^{(n)})}$ with $\widetilde{w}_{s}^{(A_{s,t}^{(n)})}$ denoting the tempered weights in cAPF as defined in Equation (5) in the main paper. Given the same input trajectory, Equation \eqref{eq:proper_weight} justifies that the particles produced by cBF and the particles produced by cAPF are properly weighted with respect to the same distribution $\pi_t(x_{0:t} \mid \theta, r_{1:t})$ \citep{liu2001monte}, and thus these two algorithms target the same distribution at each step $t$.

\section{Impact of a smaller $\epsilon$ in the ABC kernel}

To further investigate the influence of $\epsilon$ in the ABC kernel, we set $\epsilon = 0.0005$ while keeping the rest of the settings to be the same as in the Simulation Study of the main paper. The results for the three algorithms, reported via parameter RMSEs analogous to Tables 1--3 of the main paper, are shown in Tables \ref{Table:simu_4}, \ref{Table:simu_5} and \ref{Table:simu_6}. A comparison of these results with $\epsilon = 0.001$ indicates that all three algorithms overall have higher RMSEs with this smaller $\epsilon$, but ABC-PG-cAPF still consistently outperforms the other two algorithms. Moreover, these results suggest that weight degeneracy is more severe for ABC-PG-cBF and ABC-PG-cBFAS 
with the smaller $\epsilon$, as their RMSEs increased significantly more than those of ABC-PG-cAPF.
Therefore, these results provide further numerical evidence that the proposed cAPF-based algorithm is more robust to weight degeneracy.

\begin{table}[htbp!]
	\centering
	\renewcommand{\arraystretch}{1}
	\begin{tabular}{cc|cccc}
		\hline
		CV & $\phi$ & Algorithm & RMSE $\tau$ & RMSE $\phi$ & RMSE $\sigma_h^2$ \\ 
		
		\hline
		~ & ~ & ABC-PG-cBF & 0.740 & 0.098 & 0.598 \\
		10 & 0.9 & ABC-PG-cBFAS & 0.874 & 0.114 & 0.691 \\
		~ & ~ & ABC-PG-cAPF & 0.214 & 0.037 & 0.163 \\
		\hline
		
		\hline
		~ & ~ & ABC-PG-cBF & 0.901 & 0.121 & 0.596  \\
		10 & 0.95 & ABC-PG-cBFAS & 1.050 & 0.139 & 0.690  \\
		~ & ~ & ABC-PG-cAPF & 0.447 & 0.068 & 0.201  \\
		\hline
		
		\hline
		~ & ~ & ABC-PG-cBF & 1.095 & 0.150 & 0.601 \\
		10 & 0.98 & ABC-PG-cBFAS & 1.225 & 0.169 & 0.664 \\
		~ & ~ & ABC-PG-cAPF & 0.691 & 0.103 & 0.298 \\
		\hline
		
		\hline
		~ & ~ & ABC-PG-cBF & 0.824 & 0.116 & 0.582  \\
		1 & 0.9 & ABC-PG-cBFAS & 0.994 & 0.140 & 0.672 \\
		~ & ~ & ABC-PG-cAPF & 0.328 & 0.054 & 0.264 \\
		\hline
		
		\hline
		~ & ~ & ABC-PG-cBF & 1.078 & 0.153 & 0.603 \\
		1 & 0.95 & ABC-PG-cBFAS & 1.262 & 0.179 & 0.693 \\
		~ & ~ & ABC-PG-cAPF & 0.577 & 0.089 & 0.285 \\
		\hline
		
		\hline
		~ & ~ & ABC-PG-cBF & 1.239 & 0.178 & 0.597 \\
		1 & 0.98 & ABC-PG-cBFAS & 1.446 & 0.207 & 0.682 \\
		~ & ~ & ABC-PG-cAPF & 0.818 & 0.125 & 0.330 \\
		
		\hline
		~ & ~ & ABC-PG-cBF & 0.851 & 0.123 & 0.604  \\
		0.1 & 0.9 & ABC-PG-cBFAS & 1.008 & 0.146 & 0.672 \\
		~ & ~ & ABC-PG-cAPF & 0.379 & 0.061 & 0.345 \\
		\hline
		
		\hline
		~ & ~ & ABC-PG-cBF & 1.114 & 0.161 & 0.593  \\
		0.1 & 0.95 & ABC-PG-cBFAS & 1.306 & 0.189 & 0.674  \\
		~ & ~ & ABC-PG-cAPF & 0.682 & 0.105 & 0.356  \\
		\hline
		
		\hline
		~ & ~ & ABC-PG-cBF & 1.295 & 0.188 & 0.593  \\
		0.1 & 0.98 & ABC-PG-cBFAS & 1.511 & 0.219 & 0.678  \\
		~ & ~ & ABC-PG-cAPF & 0.892 & 0.136 & 0.370 \\
		\hline

	\end{tabular}
	\caption{RMSEs of the SVM parameter estimates for the three particle Gibbs algorithms, based on 100 simulated datasets with $T=100$, $N=100$, $Z_t \sim SD(1.75, 0.1, 1, 0)$, $\epsilon = 0.0005$.}
	\label{Table:simu_4}
\end{table}

\begin{table}[htbp!]
	\centering
	\renewcommand{\arraystretch}{1}
	\begin{tabular}{cc|cccc}
		\hline
		CV & $\phi$ & Algorithm & RMSE $\tau$ & RMSE $\phi$ & RMSE $\sigma_h^2$ \\ 
		
		\hline
		~ & ~ & ABC-PG-cBF & 0.765 & 0.101 & 0.620  \\
		10 & 0.9 & ABC-PG-cBFAS & 0.898 & 0.118 & 0.710 \\
		~ & ~ & ABC-PG-cAPF & 0.224 & 0.038 & 0.162 \\
		\hline
		
		\hline
		~ & ~ & ABC-PG-cBF & 0.942 & 0.126 & 0.635 \\
		10 & 0.95 & ABC-PG-cBFAS & 1.056 & 0.140 & 0.700 \\
		~ & ~ & ABC-PG-cAPF & 0.457 & 0.069 & 0.219 \\
		\hline
		
		\hline
		~ & ~ & ABC-PG-cBF & 1.091 & 0.151 & 0.596 \\
		10 & 0.98 & ABC-PG-cBFAS & 1.263 & 0.174 & 0.694 \\
		~ & ~ & ABC-PG-cAPF & 0.695 & 0.103 & 0.292 \\
		\hline
		
		\hline
		~ & ~ & ABC-PG-cBF & 0.875 & 0.123 & 0.623 \\
		1 & 0.9 & ABC-PG-cBFAS & 1.022 & 0.145 & 0.690 \\
		~ & ~ & ABC-PG-cAPF & 0.308 & 0.052 & 0.257 \\
		\hline
		
		\hline
		~ & ~ & ABC-PG-cBF & 1.090 & 0.155 & 0.616 \\
		1 & 0.95 & ABC-PG-cBFAS & 1.288 & 0.182 & 0.705 \\
		~ & ~ & ABC-PG-cAPF & 0.602 & 0.094 & 0.302 \\
		\hline
		
		\hline
		~ & ~ & ABC-PG-cBF & 1.271 & 0.183 & 0.607 \\
		1 & 0.98 & ABC-PG-cBFAS & 1.491 & 0.214 & 0.715 \\
		~ & ~ & ABC-PG-cAPF & 0.840 & 0.128 & 0.351 \\
		
		\hline
		~ & ~ & ABC-PG-cBF & 0.870 & 0.126 & 0.618 \\
		0.1 & 0.9 & ABC-PG-cBFAS & 1.024 & 0.149 & 0.693 \\
		~ & ~ & ABC-PG-cAPF & 0.346 & 0.058 & 0.336 \\
		\hline
		
		\hline
		~ & ~ & ABC-PG-cBF & 1.125 & 0.163 & 0.607 \\
		0.1 & 0.95 & ABC-PG-cBFAS & 1.348 & 0.195 & 0.706 \\
		~ & ~ & ABC-PG-cAPF & 0.682 & 0.106 & 0.357 \\
		\hline
		
		\hline
		~ & ~ & ABC-PG-cBF & 1.325 & 0.192 & 0.619 \\
		0.1 & 0.98 & ABC-PG-cBFAS & 1.555 & 0.225 & 0.715 \\
		~ & ~ & ABC-PG-cAPF & 0.889 & 0.136 & 0.367 \\
		\hline

	\end{tabular}
	\caption{RMSEs of the SVM parameter estimates for the three particle Gibbs algorithms, based on 100 simulated datasets with $T=100$, $N=100$, $Z_t \sim SD(1.7, 0.3, 1, 0)$, $\epsilon = 0.0005$.}
	\label{Table:simu_5}
\end{table}

\begin{table}[htbp!]
	\centering
	\renewcommand{\arraystretch}{1}
	\begin{tabular}{cc|cccc}
		\hline
		CV & $\phi$ & Algorithm & RMSE $\tau$ & RMSE $\phi$ & RMSE $\sigma_h^2$ \\ 
		
		\hline
		~ & ~ & ABC-PG-cBF & 0.946 & 0.122 & 0.824 \\
		10 & 0.9 & ABC-PG-cBFAS & 1.094 & 0.141 & 0.907 \\
		~ & ~ & ABC-PG-cAPF & 0.258 & 0.044 & 0.197 \\
		\hline
		
		\hline
		~ & ~ & ABC-PG-cBF & 1.109 & 0.146 & 0.790 \\
		10 & 0.95 & ABC-PG-cBFAS & 1.279 & 0.168 & 0.896 \\
		~ & ~ & ABC-PG-cAPF & 0.503 & 0.076 & 0.243 \\
		\hline
		
		\hline
		~ & ~ & ABC-PG-cBF & 1.273 & 0.174 & 0.749 \\
		10 & 0.98 & ABC-PG-cBFAS & 1.476 & 0.200 & 0.870 \\
		~ & ~ & ABC-PG-cAPF & 0.730 & 0.109 & 0.318 \\
		\hline
		
		\hline
		~ & ~ & ABC-PG-cBF & 1.042 & 0.147 & 0.772 \\
		1 & 0.9 & ABC-PG-cBFAS & 1.203 & 0.170 & 0.861 \\
		~ & ~ & ABC-PG-cAPF & 0.350 & 0.060 & 0.298 \\
		\hline
		
		\hline
		~ & ~ & ABC-PG-cBF & 1.260 & 0.178 & 0.765 \\
		1 & 0.95 & ABC-PG-cBFAS & 1.504 & 0.213 & 0.890 \\
		~ & ~ & ABC-PG-cAPF & 0.655 & 0.102 & 0.346 \\
		\hline
		
		\hline
		~ & ~ & ABC-PG-cBF & 1.430 & 0.205 & 0.753 \\
		1 & 0.98 & ABC-PG-cBFAS & 1.686 & 0.241 & 0.878 \\
		~ & ~ & ABC-PG-cAPF & 0.880 & 0.135 & 0.382 \\
		
		\hline
		~ & ~ & ABC-PG-cBF & 1.038 & 0.150 & 0.764 \\
		0.1 & 0.9 & ABC-PG-cBFAS & 1.230 & 0.179 & 0.862 \\
		~ & ~ & ABC-PG-cAPF & 0.419 & 0.069 & 0.389 \\
		\hline
		
		\hline
		~ & ~ & ABC-PG-cBF & 1.264 & 0.183 & 0.727 \\
		0.1 & 0.95 & ABC-PG-cBFAS & 1.546 & 0.224 & 0.867 \\
		~ & ~ & ABC-PG-cAPF & 0.712 & 0.112 & 0.389 \\
		\hline
		
		\hline
		~ & ~ & ABC-PG-cBF & 1.483 & 0.215 & 0.751 \\
		0.1 & 0.98 & ABC-PG-cBFAS & 1.733 & 0.252 & 0.860 \\
		~ & ~ & ABC-PG-cAPF & 0.944 & 0.145 & 0.420 \\
		\hline

	\end{tabular}
	\caption{RMSEs of the SVM parameter estimates for the three particle Gibbs algorithms, based on 100 simulated datasets with $T=100$, $N=100$, $Z_t \sim SD(1.5, -0.3, 1, 0)$, $\epsilon = 0.0005$.}
	\label{Table:simu_6}
\end{table}

\end{document}